\documentclass[12pt, draftclsnofoot, journal, letter, onecolumn]{IEEEtran} 

\usepackage{graphicx}
\usepackage{epsfig}
\usepackage{latexsym}
\usepackage{amsfonts}
\usepackage{here}
\usepackage{rawfonts}
\usepackage[latin1]{inputenc}
\usepackage[T1]{fontenc}
\usepackage{calc}
\usepackage{url}
\usepackage{enumerate}
\usepackage{color}
\usepackage[tbtags]{amsmath}
\usepackage{amssymb}
\usepackage{upref}
\usepackage{epic,eepic}
\usepackage{times}
\usepackage{dsfont}
\usepackage{comment}
\usepackage{}

\renewcommand{\Pr}{\ensuremath{\operatorname{Pr}}}

\newtheorem{proposition}{{\bf Proposition}}

\newtheorem{corollary}{{\bf Corollary}}

\newcommand{\qed}{\nobreak \ifvmode \relax \else
  \ifdim\lastskip<1.5em \hskip-\lastskip
  \hskip1.5em plus0em minus0.5em \fi \nobreak
  \vrule height0.75em width0.5em depth0.25em\fi}

\newcounter{step}
\newlength{\totlinewidth}
  {\end{list}%
  \rule{\linewidth}{1pt}}
\newcounter{substep}

  {\end{list}}

\newlength{\aligntop}
\newlength{\alignbot}
 \makeatletter
\renewenvironment{align}{%
  \vspace{\aligntop}
  \start@align\@ne\st@rredfalse\m@ne
}{%
  \math@cr \black@\totwidth@
  \egroup
  \ifingather@
    \restorealignstate@
    \egroup
    \nonumber
    \ifnum0=`{\fi\iffalse}\fi
  \else
    $$%
  \fi
  \ignorespacesafterend%
  \vspace{\alignbot}\par\noindent
} \makeatother

\IEEEoverridecommandlockouts

\begin{document}
\title{Efficient Relay Selection Scheme for Delay-Limited Non-Orthogonal
Hybrid-ARQ Relay Channels \vspace{-0.3cm}}



\author{
\authorblockN{Behrouz Maham, \emph{Member, IEEE}, Aydin Behnad, \emph{Student Member, IEEE}, and Mérouane Debbah, \emph{Senior Member, IEEE}}\\
    \thanks{
     Behrouz Maham and Aydin Behnad are with School of ECE, College of Engineering, University of Tehran,
North Karegar, Tehran 14395-515, Iran.
     Mérouane Debbah is with Alcatel-Lucent Chair on Flexible Radio, SUP{\'E}LEC,
Gif-sur-Yvette, France. Preliminary version of a portion
    of this work was appeared in \emph{Proc. IEEE Vehicular Technology Conference (VTC 2010-Fall)}.
Emails:
\protect\url{b.maham@ut.ac.ir, behnad@ut.ac.ir, merouane.debbah@supelec.fr}.}%
}

\maketitle

\begin{abstract}

We consider a half-duplex wireless relay network with hybrid-automatic retransmission request
(HARQ) and Rayleigh fading channels. In this paper, we analyze the outage probability of the multi-relay delay-limited HARQ system with opportunistic relaying scheme in decode-and-forward mode, in which the \emph{best} relay is selected to transmit the source's regenerated signal.
A simple and distributed relay selection strategy is proposed for multi-relay HARQ channels. Then, we utilize the non-orthogonal cooperative transmission between the source and selected relay for retransmitting of the source data toward the destination if needed, using space-time codes or beamforming techniques. We analyze the performance of the system. We first derive the cumulative density function (CDF) and probability density function (PDF) of the selected relay HARQ channels.
Then, the CDF and PDF are used to determine the outage probability in the $l$-th round of HARQ. The outage probability is required to compute the throughput-delay
performance of this half-duplex opportunistic relaying protocol. The packet delay constraint is represented by $L$, the maximum number of HARQ rounds. An outage is declared if the packet is unsuccessful after $L$
HARQ rounds. Furthermore, closed-form upper-bounds on outage probability are derived and subsequently are used to investigate the diversity order of the system.
Based on the derived upper-bound expressions, it is shown that the proposed schemes
achieve the full spatial diversity order of $N+1$, where $N$ is the number of potential relays.
Our analytical results are confirmed by
simulation results.
\end{abstract}

\section{Introduction}
Cooperation among devices has been considered
to provide diversity in wireless networks where fading may
significantly affect single links \cite{nos04}. Initial works have emphasized
on relaying, where a cooperator node amplifies (or decodes)
and forwards, possibly in a quantized fashion \cite{cov79}, the information
from the source node in order to help decoding at
the destination node \cite{lan04,mah09twc,mah08nov}. The achieved throughput can be
increased with the integration of cooperation and coding, i.e.,
by letting the cooperator send incremental redundancy to the
destination \cite{hun06}. In particular, it has been shown in \cite{hun06} that
coded cooperation achieves a diversity order of two, while
decode-and-forward reaches only a diversity order one, when
the transmissions of source and cooperator are orthogonal. The
capacity of cooperative networks using both the decode-and-forward
and coded cooperation has been extensively studied
\cite{hun06,kra05} for simple networks with simple medium access control
(MAC) protocols. In \cite{lin06}, a system with two transmission
phases that makes use of convolutional codes is analyzed and
characterized by means of partner choice and performance
regions. Resource allocation for space-time coded cooperative networks
has been studied in \cite{mah09eur, mah11tc}, where the analysis of bit error rate and outage probability are
also derived.
Unfortunately, in cooperative relaying the diversity gain is
increased at the expense of throughput loss due to the half-duplex
constraint at relay nodes. Different methods have been proposed to
recover this loss. In \cite{fan07}, successive relaying using repetition coding
has been introduced for a two relay wireless network with flat fading.
In \cite{tan08}, relay selection methods have been proposed for cooperative
communication with decode-and-forward (DF) relaying.

A prominent alternative to reducing the throughput loss in relay-aided
transmission mechanisms is the combination of both ARQ and
relaying. This approach would significantly reduce the half-duplex
multiplexing loss by activating ARQ for rare erroneously decoded data
packets, when they occur. Approaches targeting the joint design of
ARQ and relaying in one common protocol have recently received more
interest (see for instance \cite{nar08,qi09}). Motivated by the above suggestion,
we investigate and analyze throughput efficient cooperative transmission techniques
where both ARQ and relaying are jointly designed.
To this end,
a diversity effect can be introduced to a relay networks by simply allowing
the nodes to maintain previously received information
concerning each active message. Each time a message is retransmitted,
either from a new node or from
the same node,
every node in the relay network will increase the amount of resolution
information it has about the message. Once a node has
accumulated sufficient information it will be able to decode the
message and can act as a relay and forward the message (as in
decode-and-forward \cite{lan04,mah09iet}). This diversity effect can be viewed
as a space-time generalization of the time-diversity effect of hybrid-automatic repeat request (HARQ) as described in \cite{cai01}. Thus, the HARQ scheme which is used in this paper is a practical approach to designing
wireless ad hoc networks that exploit the spatial diversity, which is achievable with relaying. The retransmitted
packets could originate from any node that has overheard and
successfully decoded the message. Current and future wireless
networks based on packet switching use HARQ protocols at the link layer. Hence, the performance of HARQ protocols in relay channels has attracted recent research interest \cite{tab05,zha05jsac,nar08}.

In this paper, we propose an efficient HARQ multi-relay protocol which leads to full spatial diversity.
We assume a delay-limited network with the maximum number of HARQ rounds $L$, which represents the delay constraint.
The protocol uses a form of incremental redundancy HARQ transmission with assistance from the selected relay via non-orthogonal transmission in the second transmission phase
if the relay decodes the message before the destination. Note that by non-orthogonal transmission, we mean that the source node and the selected relay simultaneously retransmit the source data using space-time codes or beamforming techniques.
We introduce a distributed relay selection scheme for HARQ multi-relay networks by using acknowledgment (ACK) or non-acknowledgment (NACK) signals transmitted by destination.
Closed-form expressions are derived for the outage probability, defined as the probability of packet failure after $L$ HARQ rounds, in half-duplex. For sufficiently high SNR, we derive a simple closed-form average outage probability expression for a HARQ system with multiple cooperating branches, and it is shown that the full diversity is achievable in the proposed HARQ relay networks.
The simulations shows that the throughput of the relay channel
is significantly larger than that of direct transmission for a wide
range of signal-to-noise ratios (SNRs), target outage probabilities,
delay constraints and relay numbers.

The remainder of this paper is organized as follows:
In Section II, the system model and protocol description are given.
The performance analysis
The closed--form expressions for the outage probability and asymptotic analysis of the system are presented in Section III, which are utilized for optimizing the system.
In Section IV, the overall system performance is presented for
different numbers of relays and channel conditions, and the correctness of the analytical
formulas are confirmed by simulation results.
Conclusions are presented in Section V.

%
\emph{Notations}: The superscripts $(\cdot)^t$, $(\cdot)^H$, and $(\cdot)^*$ stand for
transposition, conjugate transposition, and element-wise
conjugation, respectively.
The expectation operation is denoted by $\mathbb{E}\{\cdot\}$.
%
%
The union and intersection of a collection of sets are denoted by $\bigcup$ and $\bigcap$, respectively.
The symbol $|x|$
is the absolute value of the scalar $x$, while
$[x]^+$ denotes $\max\{x,0\}$.
The logarithms $\log_2$ and $\log$ are the based two logarithm and the natural logarithm, respectively.

\section{System Model and Protocol Description}
\begin{figure}[e]
  \centering
  \includegraphics[width=\columnwidth]{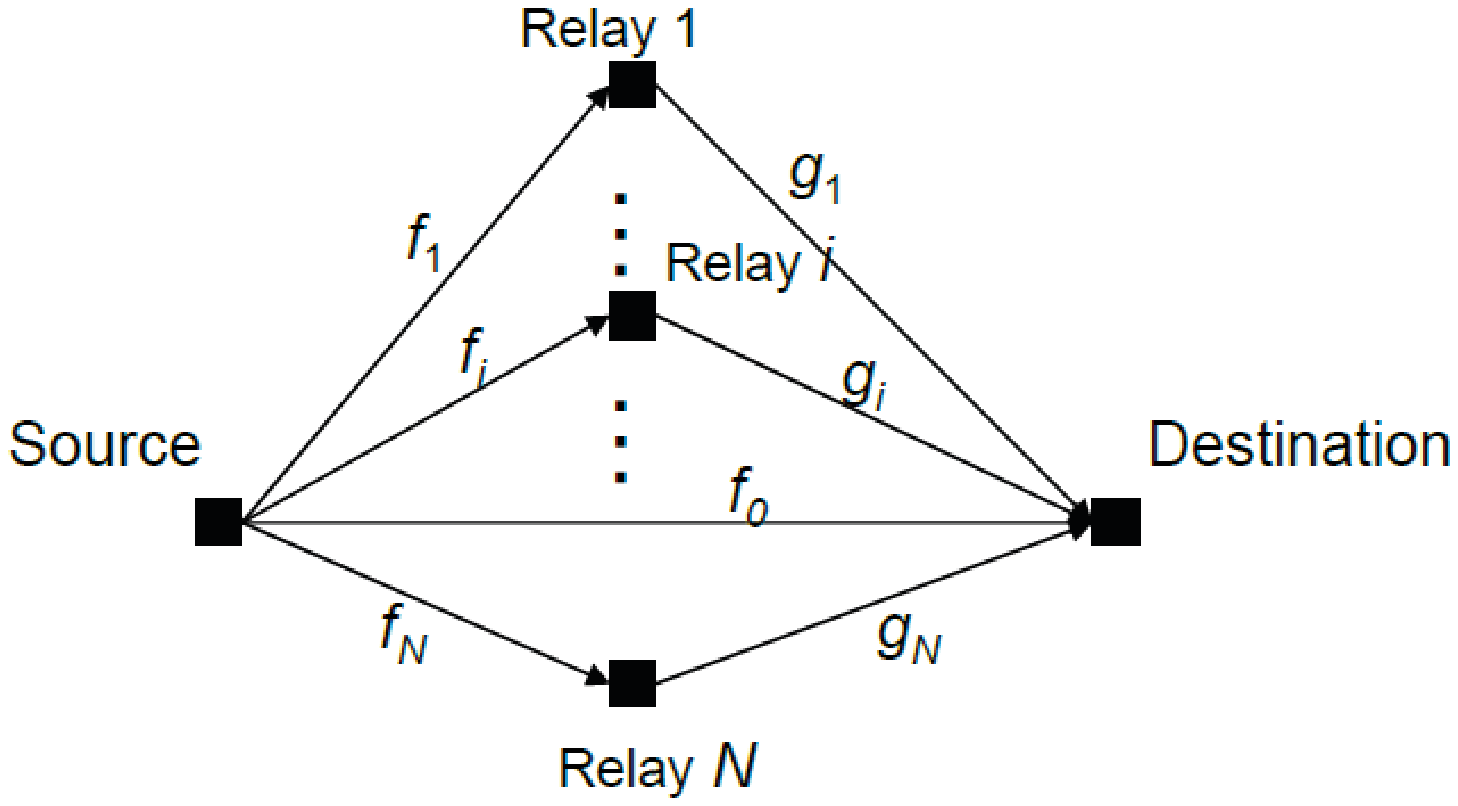}\\
  \caption{Wireless relay network consisting of a source, a destination, and \emph{N} relays.}\label{fa}
\end{figure}
\begin{figure}[e]
  \centering
  \includegraphics[width=\columnwidth]{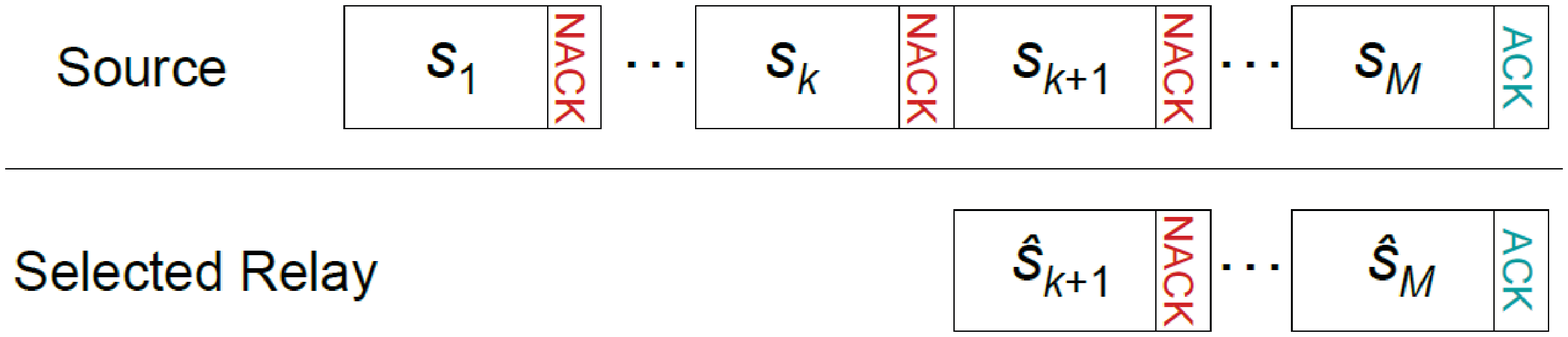}\\
  \caption{Example of HARQ protocol for relay selection system. The selected relay decodes
message after HARQ round $k$. The source and selected relay simultaneously transmit $s_l$ and $\hat{s}_l$, respectively, for all HARQ rounds $l > k$.
In this figure, the
destination decodes the message after HARQ round $M$ where $k<M\leq L$.}\label{f0}
\end{figure}
Consider a network
consisting of a source, one or more relays denoted $i=1, 2, \ldots, N$, and one
destination. The wireless relay network model is illustrated in Fig. \ref{fa}. It is assumed that each node is
equipped with a single antenna. We consider symmetric channels and
denote the source-to-destination, source-to-\emph{i}th relay, and
\emph{i}th relay-to-destination links by $f_0$, $f_i$, and $g_i$,
respectively. Suppose each link has Rayleigh fading, independent of
the others. Therefore, $f_0$, $f_i$, and $g_i$ are i.i.d. complex
Gaussian random variables with zero-mean and variances $\sigma_0^2$,
$\sigma_{f_i}^2$, and $\sigma_{g_i}^2$, respectively. As in \cite{nar08}, all
links are assumed to be long-term quasi-static wherein all
HARQ rounds of a single packet experience a single channel
realization. Subsequent packets experience independent channel
realizations. Note that such an assumption, applicable in
low-mobility environments such as indoor wireless local area
networks (WLANs), clearly reveals the gains due to HARQ
since temporal diversity is not present.

\subsection{Relay Selection Strategy}
In this paper, we use selection relaying, a.k.a. opportunistic relaying \cite{ble06b}, which selects the best relay among $N$ available relays. Inspired by the
distributed algorithm proposed in \cite{ble06b}, which uses request-to-send
(RTS) and clear-to-send (CTS) signals to select the best relay, we propose the following selection procedure for HARQ systems:
\begin{itemize}
\item In the first step, the source node broadcast its packet toward
the relays and the destination. Thus, relays can estimate their
source-to-relay channels.
\item If the destination decodes the packet correctly, the relays would not cooperate. Otherwise, relays exploit the NACK signal which is transmitted by
the destination to estimate their corresponding relay-to-destination channel.
\item The $i$th relay, $i=1,\ldots,N$ has a timer $T_i$ which its value is proportional
to the inverse of $\min\left\{|f_{i}|^2,|g_{i}|^2\right\}$.
\item The relay with maximum amount of
$\min\left\{|f_{i}|^2,|g_{i}|^2\right\}$ has a smallest $T_i$. Whenever
the first relay finished its timer, it broadcasts a flag packet
toward the other relays to make them silent and announce her as the selected relay.
\end{itemize}

Note that the process of selecting the best relay could be also done in a centralized manner by the
destination. This is feasible since the destination node should be
aware of both the backward and forward channels for coherent decoding.
Thus, the same channel information could be exploited for the
purpose of relay selection. After selecting the best relay, a feedback packet containing the index of the best relay should be sent from the destination toward the source and relay nodes.

\subsection{Transmission Strategy}
Let $s$ and $\hat{s}$ denote the transmitted signals from the source and the selected relay,
respectively. As shown in Fig. \ref{f0}, during the first HARQ round,
the relays and destination listen to the source transmit block
$s$. At the end of the transmission, the destination transmits
both the source and relays a one-bit ACK or NACK indicating, respectively, the
success or failure of the transmission. The NACK/ACK is
assumed to be received error-free and with negligible delay. Then, with the procedure given above, the best relay is selected.
As long as NACK is received after each HARQ round and the
maximum number of HARQ rounds is not reached, the source
successively transmits subsequent HARQ blocks of the same packet. As illustrated in Fig. \ref{f0}, suppose the selected relay decodes the message after
HARQ round $k$, while the destination has not yet decoded the message correctly. For all HARQ rounds $l > k$, the source and the selected relay simultaneously transmit $s$ and $\hat{s}$, respectively. For this non-orthogonal transmission, the destination can be benefited from the spatial diversity using the following methods:
\subsubsection{Space-Time Code Transmission}
The Alamouti code can be used to transmit the coded packets, hence, no interference occurs due to the simultaneous transmissions of the source and relay.
The effective coding rate after $l$ HARQ rounds is $R/l$ bps/Hz, where $R$ is the spectral efficiency (in bps/Hz) of the first HARQ round. Let $x$ and $\hat{x}$ denote the Alamouti code transmitted signals from the source and the selected relay, respectively.
The received signal $y$ at the destination can be
written as follows:
\begin{equation}\label{1}
    y=\left\{\begin{array}{cc}
                          f_0 x + g_r \hat{x}+n, & \text{if } l > k, \\
                          f_0 x +n, & \text{if } l \leq k.
                        \end{array}\right.
\end{equation}
where the index $r$ refers to the index of the selected relay and $n$ is a complex white Gaussian noise sample with
variance $N_0$.
\subsubsection{Beamforming Transmission}
An alternative way of simultaneous transmission of the coded packets is beamforming. Assuming the knowledge of channel phases of $f_0$ and $g_r$ at the source and the selected relay $r$, respectively, the transmitted message can be recovered at the destination. Moreover, using beamforming, we can achieve array gain of two, comparing to space-time code usage, in expense of higher data exchange overhead for phase estimation at the transmitters. However, ACK/NACK transmission from the destination can be exploited for the phase estimation of the channels. Thus, no signaling overhead is added when beamforming technique is used. In this case, the
received signal at the destination is given by
\begin{equation}\label{1e}
    y=\left\{\begin{array}{cc}
                          |f_0| s + |g_r| \hat{s}+n, & \text{if } l > k, \\
                          |f_0| s +n, & \text{if } l \leq k.
                        \end{array}\right.
\end{equation}

\subsection{Average Throughput}
Two definitions of throughput are considered. A frequently used metric for throughput analysis is the
long-term (LT) average throughput, given by \cite{elg06}
\begin{equation}\label{b3}
    \bar{G}_{LT}=\frac{R}{\mathbb{E}\{l\}}=\frac{R}{\sum_{l=0}^{L-1}P_{\text{out}}(l)},
\end{equation}
where $\mathbb{E}\{l\}$ is the average number of HARQ rounds spent
transmitting an arbitrary message and $P_{\text{out}}(l)$ denotes the probability that the packet is incorrectly decoded at the destination
after $l$ HARQ rounds. In the next section, we calculate closed-form solutions for the outage probability terms $P_{\text{out}}(l)$ used in \eqref{b3}.
The definition in \eqref{b3} relies on the steady-state behavior of several
message transmissions. During this time, the probabilities
$P_{\text{out}}(l)$ are assumed to be constant. This assumption is
removed by considering the delay-limited
(DL) throughput, which is the throughput of a single packet,
defined by \cite{nar08}
\begin{equation}\label{b2}
    \bar{G}_{DL}=\sum_{l=1}^{L}\frac{R}{l}\left[P_{\text{out}}(l-1)-P_{\text{out}}(l)\right],
\end{equation}
An advantage of definition \eqref{b2}, which does not resort to
long-term behavior, is the ability to track slow time variations
in the channels.

In \cite{zor03a}, the request to an automatic repeat request (ARQ) is served by the relay closest to the destination, among those that have decoded the message. However, distance-dependent relay selection does not consider the fading effect of wireless networks and leads to a maximum diversity of two. Therefore, in this work, the request to an ARQ is served by the relay with the best instantaneous channel conditions. Similar to \cite{ble06b}, we choose the relay with the maximum of $\min\left\{\gamma_{f_i},\gamma_{g_i}\right\}$, $i=1,\ldots,N$, as the best relay, where $\gamma_{f_i}=|f_{i}|^2$ and $\gamma_{g_i}=|g_{i}|^2$. We define
\begin{align}\label{2}
    \gamma_{\max}&\triangleq\min\left\{\gamma_{f_r},\gamma_{g_r}\right\}
    \nonumber\\
    &=\max\left\{\min\left\{\gamma_{f_1},\gamma_{g_1}\right\},\ldots,\min\left\{\gamma_{f_N},\gamma_{g_N}\right\}\right\}
\end{align}
where
\begin{equation}\label{3}
    r=\text{arg}\max_{i=1,\ldots,N}\left\{\min\left\{\gamma_{f_i},\gamma_{g_i}\right\}\right\}.
\end{equation}

\section{Performance Analysis}

In this section, we calculate the outage probability of the HARQ relay selection system proposed in the previous section. Besides achieving a performance metric, outage probability expression is needed in both throughput definitions in \eqref{b3} and \eqref{b2}.

Let $\chi$ denote the earliest HARQ round after which the relay
stops listening to the current message. The outage probability
for the relay channel after $l$ HARQ rounds is given by \cite{tab05}
\begin{align}\label{4}
    P_{\text{out}}(l)=&\sum_{k=1}^{l-1}P_{\text{out}}(l\,|l>k)\,\text{Pr}[\chi=k]
    \nonumber\\
    &+\sum_{k=l}^{L}P_{\text{out}}(l\,|l\leq k)\,\text{Pr}[\chi=k].
\end{align}
To compute $\text{Pr}[\chi=k]$, the mutual information between
source and relay for each HARQ round is given by
\begin{equation}\label{5}
    I_{f_r}=\log_2\left(1+\frac{P}{N_0}\gamma_{f_r}\right),
\end{equation}
where $P$ is the average transmit power from the source and $\gamma_{f_r}$ is an exponentially distributed random variable with mean $\sigma_{f_r}^2$. For $k = 1, \ldots , l - 1$, $\chi=k$ if the message
is successfully decoded by the relay at the $k$th HARQ round, and we have
\begin{align}\label{6}
    \text{Pr}[\chi=k]&=\text{Pr}[(k-1)I_{f_r}<R,\,k \, I_{f_r}>R]
    \nonumber\\
    &=\text{Pr}[(k-1)I_{f_r}<R]-\text{Pr}[k \,I_{f_r}<R]
    \nonumber\\
    &=\text{Pr}[\gamma_{f_r}<\mu_{k-1}]-\text{Pr}[\gamma_{f_r}<\mu_{k}],
\end{align}
where
\begin{equation}\label{7}
    \mu_{k}=\frac{N_0}{P}\left(2^{R/k}-1\right).
\end{equation}

For $k = l, \ldots , L$, $\chi=k$ if the relay did not decode the message successfully after $(l-1)$ HARQ rounds, and thus, we have
\begin{align}\label{8}
    \text{Pr}[\chi=k]&=\text{Pr}[(l-1)I_{f_r}<R]
=\text{Pr}[\gamma_{f_r}<\mu_{l-1}].
\end{align}
From \eqref{6} and \eqref{8}, $\text{Pr}[\chi=k]$ can be calculated as
\begin{align}\label{9}
    \text{Pr}[\chi=k]&=\left\{\begin{array}{cc}
                          \text{Pr}[\gamma_{f_r}<\mu_{k-1}]-\text{Pr}[\gamma_{f_r}<\mu_{k}], & \text{if } k<l, \\
                          \text{Pr}[\gamma_{f_r}<\mu_{l-1}], & \text{if } k\geq l.
                        \end{array}\right.
\end{align}

\subsection{Exact Outage Probability}
Since the index $r$ given in \eqref{3} is dependent on channels, $\gamma_{f_r}$ and $\gamma_{g_r}$ are \emph{not independent} for $N>1$. Thus, obtaining a closed-form for PDF is not straightforward. As it is seen from \eqref{9}, for computing $\text{Pr}[\chi=k]$, the CDF of random variable $\gamma_{f_r}$ is required.
In the following, the CDF of the random variable $\gamma_{f_r}$ is derived.

\begin{proposition}\label{a}
Let $\gamma_{f_{i}}$ and $\gamma_{g_{i}}$, $i=1,\ldots,N$, be independent exponential random variables with means $\sigma^{2}_{f_{i}}$ and $\sigma^{2}_{g_{i}}$, respectively.
The CDF and PDF of $\gamma_{f_r}$, where $r$ is defined as \eqref{3}, are given by \eqref{P1} and \eqref{P2}, respectively.

\begin{align}\label{P1}
    \textnormal{F}_{\gamma_{f_{r}}}(\gamma)&=\prod_{i=1}^{N}\Big(1-e^{-(\frac{1}{\sigma^{2}_{f_{i}}}+\frac{1}{\sigma^{2}_{g_{i}}})\gamma}\Big)-\sum_{j=1}^{N}\frac{1}{\sigma^{2}_{g_{j}}}e^{-\frac{\gamma}{\sigma^{2}_{f_{j}}}}\int_{0}^{\gamma}e^{-\frac{\beta}{\sigma^{2}_{g_{j}}}}\prod_{\underset{i\neq j}{i=1}}^{N}\big(1-e^{-(\frac{1}{\sigma^{2}_{f_{i}}}+\frac{1}{\sigma^{2}_{g_{i}}})\beta}\big)d\beta\hspace{2mm},\hspace{2mm}\gamma \geqslant 0 ,
\end{align}
\begin{align}\label{P2}
   \textnormal{f}_{\gamma_{f_{r}}}(\gamma)&=\sum_{j=1}^{N}\Big(\frac{1}{\sigma^{2}_{f_{j}}}
e^{-(\frac{1}{\sigma^{2}_{f_{j}}}+\frac{1}{\sigma^{2}_{g_{j}}})\gamma}\prod_{\underset{i\neq j}{i=1}}^{N}\big(1-e^{-(\frac{1}{\sigma^{2}_{f_{i}}}+\frac{1}{\sigma^{2}_{g_{i}}})\gamma}\big) \nonumber\\
&\,\,\,+\frac{1}{\sigma^{2}_{f_{j}}\sigma^{2}_{g_{j}}}e^{-\frac{\gamma}{\sigma^{2}_{f_{j}}}}\int_{0}^{\gamma}e^{-\frac{\beta}{\sigma^{2}_{g_{j}}}}\prod_{\underset{i\neq j}{i=1}}^{N}\big(1-e^{-(\frac{1}{\sigma^{2}_{f_{i}}}+\frac{1}{\sigma^{2}_{g_{i}}})\beta}\big)d\beta\Big)\hspace{2mm},\hspace{2mm}\gamma \geqslant 0.
\end{align}
\end{proposition}

\begin{proof}
The proof is given in Appendix I. 
\end{proof}

It is noteworthy that the integrals in \eqref{P1} and \eqref{P2} can be easily calculated, as all terms in the expansion of integrands  are of the exponential form.

\begin{corollary}
When $\sigma^{2}_{f_{i}}=\sigma^{2}_{f}$ and $\sigma^{2}_{g_{i}}=\sigma^{2}_{g}$ for all $i=1,2,...,N$, CDF and PDF of $\gamma_{f_{r}}$ are respectively simplified as follows:
\begin{align}\label{P3}
\textnormal{F}_{\gamma_{f_{r}}}(\gamma)&=(1-e^{-(\frac{1}{\sigma^{2}_{f}}+\frac{1}{\sigma^{2}_{g}})\gamma}\big)^{N}\nonumber \\
&-\frac{N\sigma^{2}_{f}}{\sigma^{2}_{f}+{\sigma^{2}_{g}}}e^{-\frac{\gamma}{\sigma^{2}_{f}}}\,\mathcal{B}\!\left(1-e^{-(\frac{1}{\sigma^{2}_{f}}+\frac{1}{\sigma^{2}_{g}})\gamma};N,\frac{\sigma^{2}_{f}}{\sigma^{2}_{f}+{\sigma^{2}_{g}}}\right) \\
\textnormal{f}_{\gamma_{f_{r}}}(\gamma)&=\frac{N}{\sigma^{2}_{f}}e^{-(\frac{1}{\sigma^{2}_{f}}+\frac{1}{\sigma^{2}_{g}})\gamma}(1-e^{-(\frac{1}{\sigma^{2}_{f}}+\frac{1}{\sigma^{2}_{g}})\gamma}\big)^{N-1}\nonumber \\
&+\frac{N}{\sigma^{2}_{f}+{\sigma^{2}_{g}}}e^{-\frac{\gamma}{\sigma^{2}_{f}}}\,\mathcal{B}\!\left(1-e^{-(\frac{1}{\sigma^{2}_{f}}+\frac{1}{\sigma^{2}_{g}})\gamma};N,\frac{\sigma^{2}_{f}}{\sigma^{2}_{f}+{\sigma^{2}_{g}}}\right)
\end{align}
where $\mathcal{B}(x;a,b)=\int_{0}^{x}t^{a-1}(1-t)^{b-1}dt$ is the incomplete beta function \cite{gra96}.
\end{corollary}

\begin{corollary}
For high values of SNR, i.e. when $\gamma=\mu_{k}\rightarrow 0$, the closed form solution for \eqref{P1} can be obtained as
\begin{equation}
\textnormal{F}_{\gamma_{f_{r}}}(\gamma)\cong \gamma^{N}\big(\prod_{i=1}^{N}(\frac{1}{\sigma^{2}_{f_{i}}}+\frac{1}{\sigma^{2}_{g_{i}}})\big)\big(\frac{1}{N}\sum_{i=1}^{N}\frac{\sigma^{2}_{g_{i}}}{\sigma^{2}_{f_{i}}+\sigma^{2}_{g_{i}}}\big) .
\end{equation}
\end{corollary}
\vspace{.5cm}
From \eqref{P1}, $\text{Pr}[\chi=k]$ in \eqref{9} can be written as
\begin{align}\label{9f}
    \text{Pr}[\chi=k]&=\left\{\begin{array}{cc}
                          \textnormal{F}_{\gamma_{f_{r}}}(\mu_{k-1})-\textnormal{F}_{\gamma_{f_{r}}}(\mu_{k}), & \text{if } k<l, \\
                          \textnormal{F}_{\gamma_{f_{r}}}(\mu_{l-1}), & \text{if } k\geq l.
                        \end{array}\right.
\end{align}

Next, the conditional probabilities $P_{\text{out}}(l\,|l> k)$ and $P_{\text{out}}(l\,|l\leq k)$ in \eqref{6} will be calculated.
After correct decoding of the source packet at the relay, the relay helps the source by simultaneous transmission
according to the Alamouti code. Hence, assuming the relay transmits the same power $P$ as the source, the mutual information of the effective channel is given by
\begin{equation}\label{17}
    I_{s,r,d}=\log_2\left(1+\frac{P}{N_0}\gamma_{f_0}+\frac{P}{N_0}\gamma_{g_r}\right).
\end{equation}
Let $I_{\text{tot},k,l}$ denote the total mutual information
accumulated at the destination after $l$ HARQ rounds and when
$\chi=k$. For $k < l$, the relay listens for $k$ HARQ rounds and
transmits the message simultaneously with the source using
the Alamouti code for the remaining $(l - k)$ HARQ rounds.
For $k \geq l$, the relay does not help the source during the $l$ HARQ rounds. Hence,
\begin{align}\label{18}
    I_{\text{tot},k,l}&=\left\{\begin{array}{cc}
                          k\,I_{f_0}+(l-k)\,I_{s,r,d}, & \text{if } k=1, \ldots , l - 1, \\
                          l I_{f_0}, & \text{if } k = l, \ldots , L,
                        \end{array}\right.
\end{align}
where $I_{f_0}$ is the mutual information between the source and destination at each HARQ round and can be written as
    $I_{f_0}=\log_2\left(1+\frac{P}{N_0}\gamma_{f_0}\right)$.

Therefore, for $k \geq l$, we have
\begin{align}\label{20}
    P_{\text{out}}(l\,|l\leq k)&=\text{Pr}[l I_{f_0}<R]
    =  1-\exp\left(\frac{-\mu_{l}}{\sigma_{f_0}^2}\right).
\end{align}


From \eqref{18}, the conditional probability $P_{\text{out}}(l\,|l> k)$ can be calculated as
\begin{align}\label{21}
    &P_{\text{out}}(l\,|l> k)=\text{Pr}[I_{\text{tot},k,l}<R]
    \nonumber\\
    &=\text{Pr}\!\left\{\log_2\!\left[\!\left(1\!+\!\frac{P}{N_0}\gamma_{f_0}\right)^{\!k}
    \!\left(1\!+\!\frac{P}{N_0}\gamma_{f_0}\!+\!\frac{P}{N_0}\gamma_{g_r}\right)^{\!l-k}\!\right]\!<R\!\right\}
    \nonumber\\
    &=\text{Pr}\left\{\gamma_{g_r}<\frac{2^{R/(l-k)}}{\frac{P}{N_0}\,\left(1+\frac{P}{N_0}\gamma_{f_0}\right)
    ^{k/(l-k)}}-\gamma_{f_0}
    -\frac{N_0}{P}\right\}
    \nonumber\\
    &=\int_{\gamma_{f_0}=0}^{\mu_l}\int_{\gamma_{g_r}=0}^{\beta(\gamma_{f_0})}
    \frac{e^{-\frac{\gamma_{f_0}}{\sigma^2_{f_0}}}}{\sigma^2_{f_0}}\, \textnormal{f}_{\gamma_{g_r}}(\gamma_{g_r})\, d\gamma_{f_0}\,d\gamma_{g_r}\triangleq \Upsilon(l,k),
\end{align}
where $\beta(\gamma_{f_0})=\frac{2^{R/(l-k)}N_0}{P\left(1+\frac{P}{N_0}\gamma_{f_0}\right)^{k/(l-k)}}-\gamma_{f_0}-\frac{N_0}{P}$.
Due to symmetry, the PDF of random variable $\gamma_{g_r}$, i.e.,
$\textnormal{f}_{\gamma_{g_r}}(\gamma)$, is same as the PDF of random variable $\gamma_{f_r}$, with perhaps different mean. Thus, the PDF of $\gamma_{g_r}$ can be found by the derivation of
$\text{Pr}\{\gamma_{g_r}<\gamma\}$ in \eqref{13} as
\begin{align}\label{22q}
    &\textnormal{f}_{\gamma_{g_r}}(\gamma)=\sum_{j=1}^{N}\Big(\frac{1}{\sigma^{2}_{g_{j}}}
e^{-(\frac{1}{\sigma^{2}_{f_{j}}}+\frac{1}{\sigma^{2}_{g_{j}}})\gamma}\prod_{\underset{i\neq j}{i=1}}^{N}\big(1-e^{-(\frac{1}{\sigma^{2}_{f_{i}}}+\frac{1}{\sigma^{2}_{g_{i}}})\gamma}\big) \nonumber\\
&\,\,\,+\frac{1}{\sigma^{2}_{f_{j}}\sigma^{2}_{g_{j}}}e^{-\frac{\gamma}{\sigma^{2}_{g_{j}}}}
\int_{0}^{\gamma}e^{-\frac{\beta}{\sigma^{2}_{f_{j}}}}\prod_{\underset{i\neq j}{i=1}}^{N}\big(1-e^{-(\frac{1}{\sigma^{2}_{f_{i}}}+\frac{1}{\sigma^{2}_{g_{i}}})\beta}\big)d\beta\Big)\hspace{2mm},\hspace{2mm}\gamma \geqslant 0.
\end{align}
By substituting $\textnormal{f}_{\gamma_{g_r}}(\gamma)$ from \eqref{P2} into
\eqref{21}, $P_{\text{out}}(l\,|l> k)$ is
obtained. Therefore, using \eqref{9}, \eqref{20}, and \eqref{21},
the outage probability in the $l$th stage of HARQ process can be
achieved as
\begin{align}\label{21o}
    &P_{\text{out}}(l)= \sum_{k=1}^{l-1}(\textnormal{F}_{\gamma_{f_{r}}}(\mu_{k-1})-\textnormal{F}_{\gamma_{f_{r}}}(\mu_{k}))\Upsilon(l,k)
    +\sum_{k=l}^{L}\textnormal{F}_{\gamma_{f_{r}}}(\mu_{l-1})\left(1-e^{\frac{-\mu_l}{\sigma_{f_0}^2}}\right),
\end{align}
where $\Upsilon(l,k)$ is defined in \eqref{21}.

\subsection{Approximate Outage Probability}
In the previous subsection, we were able to derive the outage probability in the $l$-th round of HARQ. However, a 
triple integral should be solved to get $\Upsilon(l,k)$, and also $\textnormal{F}_{\gamma_{f_{r}}}(\mu_{k-1})$ is in an integral form. To give an insight on the diversity and find ways of optimizing the system, in this subsection, we try to find a simpler solution for the outage probability.
In the following, an approximation of the CDF of the random variable $\gamma_{f_r}$ is derived.

\begin{proposition}\label{a}
Let $\gamma_{f_i}$ and $\gamma_{g_i}$, $i=1,\ldots,N$, be set of independence exponential random variables
with mean $\sigma^2_{f_i}=\sigma^2_{g_i}=\sigma^2_{i}$.
The cumulative density function of $\gamma_{f_r}$, where $r$ is defined as \eqref{3}, can be approximated as
\begin{align}\label{14q}
    &\textnormal{F}_{\gamma_{f_{r}}}(\gamma)
    \approx 1-\sqrt{1-\prod_{i=1}^N\left(1-e^{-\frac{2\gamma}{\sigma^2_{i}}}\right)}.
\end{align}
\end{proposition}

\begin{proof}
The proof is given in Appendix II. 
\end{proof}
In Fig. \ref{f1}, we have compared the approximated PDF of $\gamma_{f_r}$, which is obtained by the derivation of CDF in \eqref{14q}, with the simulated PDF of $\gamma_{f_r}$. As it can be seen from Fig. \ref{f1}, for the case of a single-relay network ($N=1$), the analytical and simulated results have the same performance. This is because of the fact that the independence assumption for $\gamma_{f_i}$ and $\gamma_{g_i}$ becomes valid for $N=1$, and the approximation in \eqref{14q} turns into equality. For the opportunistic relaying case, i.e., $N>1$, it can be seen that the analytical curves appropriately approximate the simulation result.
\begin{figure}[e]
  \centering
  \includegraphics[width=\columnwidth]{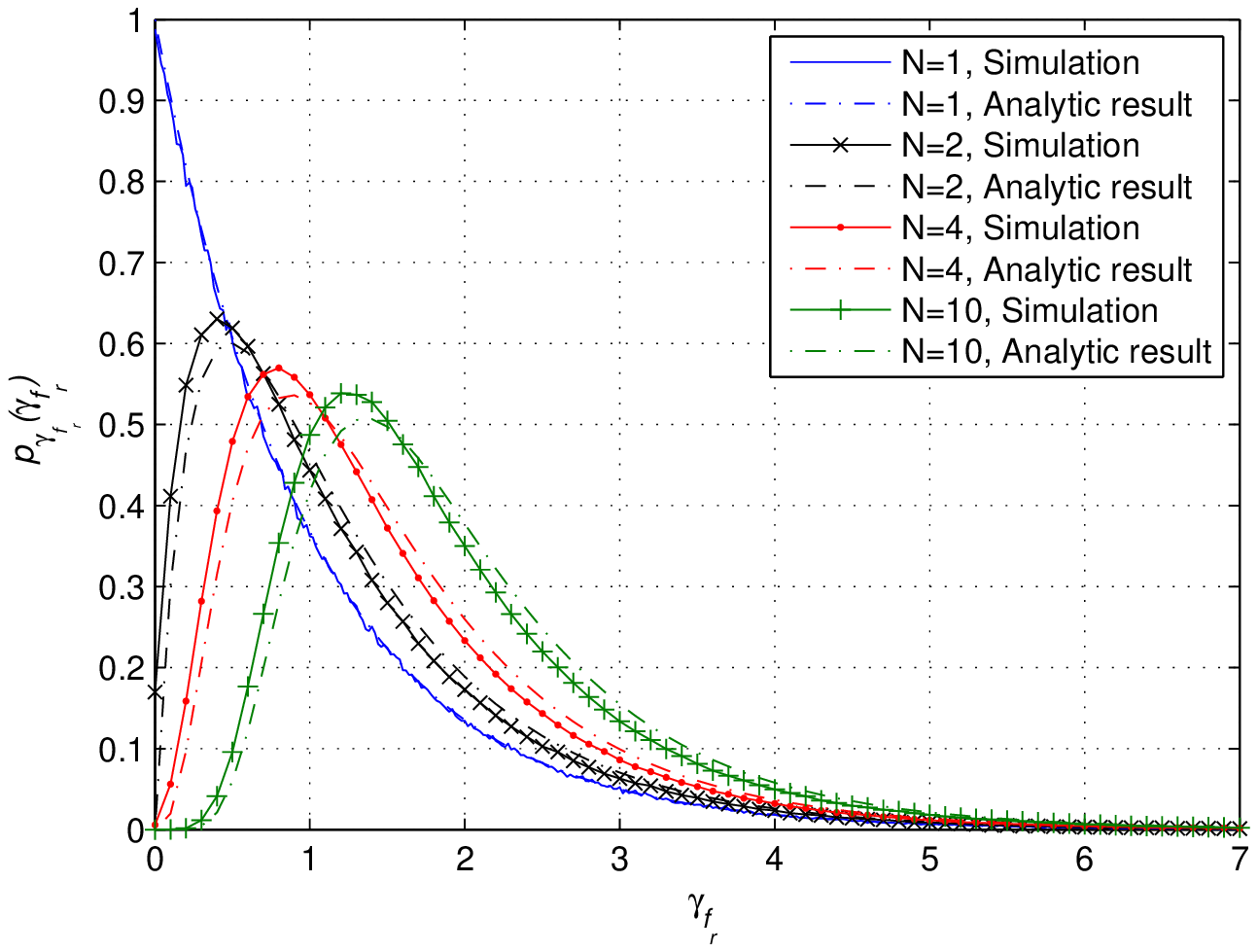}\\
  \caption{Comparison of the PDF of the received SNR at the selected relay $r$ in a network with $N$ relays.}\label{f1}
\end{figure}

From \eqref{14q}, $\text{Pr}[\chi=k]$ in \eqref{9} can be approximated as
\begin{align}\label{15q}
    \text{Pr}[\chi=k]&\approx \sqrt{1-\prod_{i=1}^N\left(1-e^{-\frac{2\mu_{k}}{\sigma^2_{i}}}\right)}
    \nonumber\\
    &-\sqrt{1-\prod_{i=1}^N
    \left(1-e^{-\frac{2\mu_{k-1}}{\sigma^2_{i}}}\right)}\triangleq \Omega_1(k),
\end{align}
for $k<l$, and
\begin{align}\label{16q}
    &\text{Pr}[\chi=k]\approx
    1-\sqrt{1-\prod_{i=1}^N\left(1-e^{-\frac{2\mu_{l-1}}{\sigma^2_{i}}}\right)}\triangleq \Omega_2(l),
\end{align}
for $k\geq l$.


Due to symmetry, the PDF of random variable $\gamma_{g_r}$, i.e.,
$\textnormal{f}_{\gamma_{g_{r}}}(\gamma)$, is same as the PDF of random variable $\gamma_{g_r}$, with perhaps different mean. Thus, the closed-form approximation for PDF of $\gamma_{g_r}$ can be found by the derivation of
$\textnormal{F}_{\gamma_{f_{r}}}(\gamma)$ in \eqref{14q} as
\begin{align}\label{22q}
    &\textnormal{f}_{\gamma_{g_r}}(\gamma)\approx\frac{1}{\sqrt{1-\displaystyle\prod_{i=1}^N\left(1-e^{-\frac{2\gamma}{\sigma^2_{i}}}\right)}}
    \sum_{i=1}^{N}\frac{    e^{-\frac{2\gamma}{\sigma^2_{i}}}     }{\sigma^2_{i}}\prod_{\underset{j\neq
    i}{j=1}}^N\!\!\left(1\!-e^{\!-\frac{2\gamma}{\sigma^2_{j}}}\right).
\end{align}
By substituting $\textnormal{f}_{\gamma_{g_r}}(\gamma)$ from \eqref{22q} into
\eqref{21}, $P_{\text{out}}(l\,|l> k)$ can be approximated as
\begin{align}\label{21q}
    &P_{\text{out}}(l\,|l> k)\approx\int_{\gamma_{f_0}=0}^{\mu_l}\int_{\gamma_{g_r}=0}^{\beta(\gamma_{f_0})}
    \frac{e^{-\frac{\gamma_{f_0}}{\sigma^2_{f_0}}}}{\sigma^2_{f_0}}\, \textnormal{f}_{\gamma_{g_r}}(\gamma_{g_r})\, d\gamma_{f_0}\,d\gamma_{g_r}\triangleq \Upsilon_2(l,k).
\end{align}
Therefore, using \eqref{9}, \eqref{20}, and \eqref{21q},
the outage probability in the $l$th stage of HARQ process can be
achieved as
\begin{align}\label{21oq}
    &P_{\text{out}}(l)\approx \sum_{k=1}^{l-1}\Omega_1(k)\Upsilon_2(l,k)+\sum_{k=l}^{L}\Omega_2(l)\left(1-e^{\frac{-\mu_l}{\sigma_{f_0}^2}}\right).
\end{align}

\subsection{Upper-Bound on Outage Probability}
For calculating the minimum diversity gain of HARQ wireless relay networks when selection strategy in \eqref{3} is used, it is enough to derive an upper-bound on the outage probability $P_{\text{out}}(l)$.

The random variable $\gamma_{f_r}$, which is corresponding the
source-relay channel of the selected relay, can be bounded as
\begin{equation}\label{4c}
    \gamma_{\max}\leq \gamma_{f_r}\leq \gamma^s_{\max},
\end{equation}
where $\gamma_{\max}$ is given in \eqref{2}
and $\gamma^s_{\max}$ is defined as
\begin{equation}\label{4a}
    \gamma^s_{\max}=\max_{i=1,\ldots,N}\left\{\gamma_{f_i}\right\}.
\end{equation}
The CDF of $\gamma^s_{\max}$ can be written as
\begin{align}\label{6ab}
    \text{Pr}\{\gamma^s_{\max}<\gamma\}&=\text{Pr}\{\gamma_{f_1}<\gamma,\gamma_{f_2}<\gamma,\ldots,
    \gamma_{f_N}<\gamma\}
    \nonumber\\
    &=\prod_{i=1}^N\left(1-e^{-\frac{\gamma}{\sigma^2_{f_i}}}\right).
\end{align}
Thus, it is easy to show that
the CDF of $\gamma_{f_r}$ can be bounded as
\begin{align}\label{18d}
    \text{Pr}\{\gamma^s_{\max}<\gamma\}&\leq\text{Pr}\{\gamma_{f_r}<\gamma\}\leq\text{Pr}\{\gamma_{\max}<\gamma\}.
\end{align}

Therefore, combining \eqref{9} and \eqref{18d}, an upper-bound on
$\text{Pr}[\chi=k]$ will be obtained as follows
\begin{align}\label{10c}
    \text{Pr}[\chi=k]\leq\text{Pr}[\gamma_{\max}<\mu_{k-1}]-\text{Pr}[\gamma^s_{\max}<\mu_{k}],
\end{align}
for $k<l$. From \eqref{11}, \eqref{6ab}, and \eqref{10c}, $\text{Pr}[\chi=k]$ for $k<l$ can be calculated as
\begin{align}\label{10cc}
    \text{Pr}[\chi=k]&\leq
    \prod_{i=1}^N\!\left(1-e^{-\mu_{k-1}\left(\frac{1}{\sigma^2_{f_i}}+\frac{1}{\sigma^2_{g_i}}\right)}\right)
    \!-\!\prod_{i=1}^N\!\left(1-e^{-\frac{\mu_{k}}{\sigma^2_{f_i}}}\right)
\nonumber\\
&    \triangleq \Lambda_1(k).
\end{align}

For $k\geq l$, by combining \eqref{9}, \eqref{11}, and \eqref{18d}, we have
\begin{align}\label{11cc}
    &\text{Pr}[\chi=k]\leq\text{Pr}[\gamma_{\max}<\mu_{l-1}]
    \nonumber\\
    &=\prod_{i=1}^N\left(1-e^{-\mu_{l-1}\left(\frac{1}{\sigma^2_{f_i}}+\frac{1}{\sigma^2_{g_i}}\right)}\right)\triangleq \Lambda_2(l).
\end{align}

Next, $P_{\text{out}}(l\,|l> k)$ in \eqref{21} can be upper-bounded as
\begin{align}\label{21j}
    &P_{\text{out}}(l\,|l> k)
    \nonumber\\
    &\leq\text{Pr}\left\{\gamma_{\max}<\frac{2^{R/(l-k)}}{\frac{P}{N_0}\,\left(1+\frac{P}{N_0}\gamma_{f_0}\right)
    ^{k/(l-k)}}-\gamma_{f_0}
    -\frac{N_0}{P}\right\}
    \nonumber\\
    &=\int_{\gamma_{f_0}=0}^{\mu_l}\int_{\gamma_{\max}=0}^{\beta(\gamma_{f_0})}
    \frac{e^{-\frac{\gamma_{f_0}}{\sigma^2_{f_0}}}}{\sigma^2_{f_0}}\, p_{\gamma_{\max}}(\gamma_{\max})\, d\gamma_{f_0}\,d\gamma_{\max}.
\end{align}

The PDF of random variable $\gamma_{\max}$, i.e.,
$\textnormal{f}_{\gamma_{\max}}(\gamma)$ can be found by the derivative of
$\text{Pr}\{\gamma_{\max}<\gamma\}$ in \eqref{11}. Thus, we have
\begin{align}\label{14hh}
    \textnormal{f}_{\gamma_{\max}}(\gamma)&=
    \sum_{i=1}^{N}   \left(\frac{1}{\sigma^2_{f_i}}+\frac{1}{\sigma^2_{g_i}}\right) e^{-\gamma \left(\frac{1}{\sigma^2_{f_i}}+\frac{1}{\sigma^2_{g_i}}\right)}
    \nonumber\\
    &\times \prod_{\underset{j\neq
    i}{j=1}}^N \left(1-e^{-\gamma\left(\frac{1}{\sigma^2_{f_j}}+\frac{1}{\sigma^2_{g_j}}\right)}\right).
\end{align}
Therefore, by substituting $\text{Pr}[\chi=k]$ from \eqref{10cc} and \eqref{11cc}, and $P_{\text{out}}(l\,|l\leq k)$ and $P_{\text{out}}(l\,|l> k)$ from \eqref{20} and \eqref{21j}, respectively, in \eqref{4}, an upper-bound on outage probability the $l$th stage of HARQ process, i.e., $P_{\text{out}}(l)$ can be achieved.

A tractable definition of the diversity gain is \cite[Eq.
(1.19)]{jaf05}
\begin{align}\label{14h}
G_d=-\lim_{\rho\rightarrow\infty}\frac{\log
\left(P_{\text{out}}\right)}{\log \left(\rho\right)},
\end{align}
where $\rho=\frac{P}{N_0}$. Thus, in the following, we investigate the
asymptotic behavior and diversity order of $P_{\text{out}}(l)$ in \eqref{4}.

From \eqref{14hh}, an upper-bound for $p_{\gamma_{\max}}(\gamma)$ can be found as
\begin{align}\label{14k}
    \textnormal{f}_{\gamma_{\max}}(\gamma)&\leq
    N \gamma^{N-1}\prod_{i=1}^N \left(\frac{1}{\sigma^2_{f_i}}+\frac{1}{\sigma^2_{g_i}}\right),
\end{align}
which is a tight bound when $\gamma\rightarrow 0$. Note that in high SNR scenario, the the behavior of the fading distribution around zero is important (see, e.g., \cite{rib05}).

Using \eqref{14k} and the fact that the exponential distribution is a
decreasing function of $\gamma_{f_0}$, $P_{\text{out}}(l\,|l> k)$ in \eqref{21j} can be further upper-bounded as
\begin{align}\label{kp}
    &P_{\text{out}}(l\,|l> k)
    \nonumber\\
    &\leq\int_{\gamma_{f_0}=0}^{\mu_l}\int_{\gamma_{\max}=0}^{\beta(\gamma_{f_0})}
    \frac{1}{\sigma^2_{f_0}}\, \gamma_{\max}^{N-1}\prod_{i=1}^N \left(\frac{1}{\sigma^2_{f_i}}+\frac{1}{\sigma^2_{g_i}}\right)\, d\gamma_{f_0}\,d\gamma_{\max}
    \nonumber\\
    &
    \leq\frac{\mu_l}{\sigma^2_{f_0}}\, \mu_{l-k}^{N}\prod_{i=1}^N \left(\frac{1}{\sigma^2_{f_i}}+\frac{1}{\sigma^2_{g_i}}\right)\triangleq \Psi(l,k).
\end{align}
Combining \eqref{4}, \eqref{20}, \eqref{10cc}, \eqref{11cc}, and \eqref{21p}, a closed-form upper-bound for the outage probability after $l$ HARQ round can be obtained as
\begin{align}\label{21p}
    &P_{\text{out}}(l)\leq \sum_{k=1}^{l-1}\Lambda_1(k)\Psi(l,k)+\sum_{k=l}^{L}\Lambda_2(l)\left(1-e^{\frac{-\mu_l}{\sigma_{f_0}^2}}\right).
\end{align}

Furthermore, using \eqref{15q}, and \eqref{16q}, another closed-form approximation for $P_{\text{out}}(l)$ can be obtained as
\begin{align}\label{21w}
    &P_{\text{out}}(l)\approx \sum_{k=1}^{l-1}\Omega_1(k)\Psi(l,k)+\sum_{k=l}^{L}\Omega_2(l)\left(1-e^{\frac{-\mu_l}{\sigma_{f_0}^2}}\right).
\end{align}

\begin{proposition}\label{a}
Assuming a HARQ system with $N$ potential relays nodes, the relay
selection strategy based on \eqref{3} can achieve the full diversity
order of $N+1$.
\end{proposition}
\begin{proof}
The proof is given in Appendix III. 
\end{proof}

\section{Numerical Analysis}
In this section, the performance of the proposed relay-selection HARQ system is studied through numerical results.
We used the equal power allocation among the source and the selected relay. Assume the relays and the destination have the same value of noise power, and
all the links have unit-variance Rayleigh flat fading, i.e., $\sigma^2_{f_i}=\sigma^2_{g_i}=\sigma^2_{f_0}=1$. It is also assumed that rate $R$ is normalized to 1.
We compare the transmit SNR
$\frac{P}{N_0}$ versus outage probability performance. 
The block fading model is used, in which channel coefficients changed randomly in time to isolate the benefits of spatial diversity. The simulation result is averaged over 3'000'000 transmitted symbols (channel realization trials).

Fig. \ref{f3} confirms that the analytical results attained in
Section~III for the outage probability have an accurate performance as the simulation results. We consider the maximum number of HARQ rounds to be $L=5$. The outage probability at the 2nd HARQ round, i.e., $P_{\text{out}}(l=2)$, is compared for two different number of relays $N=2,4$. One can see the approximate outage probability derived in \eqref{21oq} has the similar performance as the simulated curved for all values of SNR. In addition, the closed-form outage probability expression in \eqref{21w} well approximates the simulated results, especially in medium and high SNR conditions. Furthermore, Fig. \ref{f3} shows that the upper-bound expression in \eqref{21p} is a tight upper-bound.
The asymptotic outage probability derived in \eqref{60} is also depicted in Fig. \ref{f3} which confirms the full-diversity order of the proposed scheme.

\begin{figure}[e]
  \centering
  \vspace{-.3cm}
  \includegraphics[width=\columnwidth]{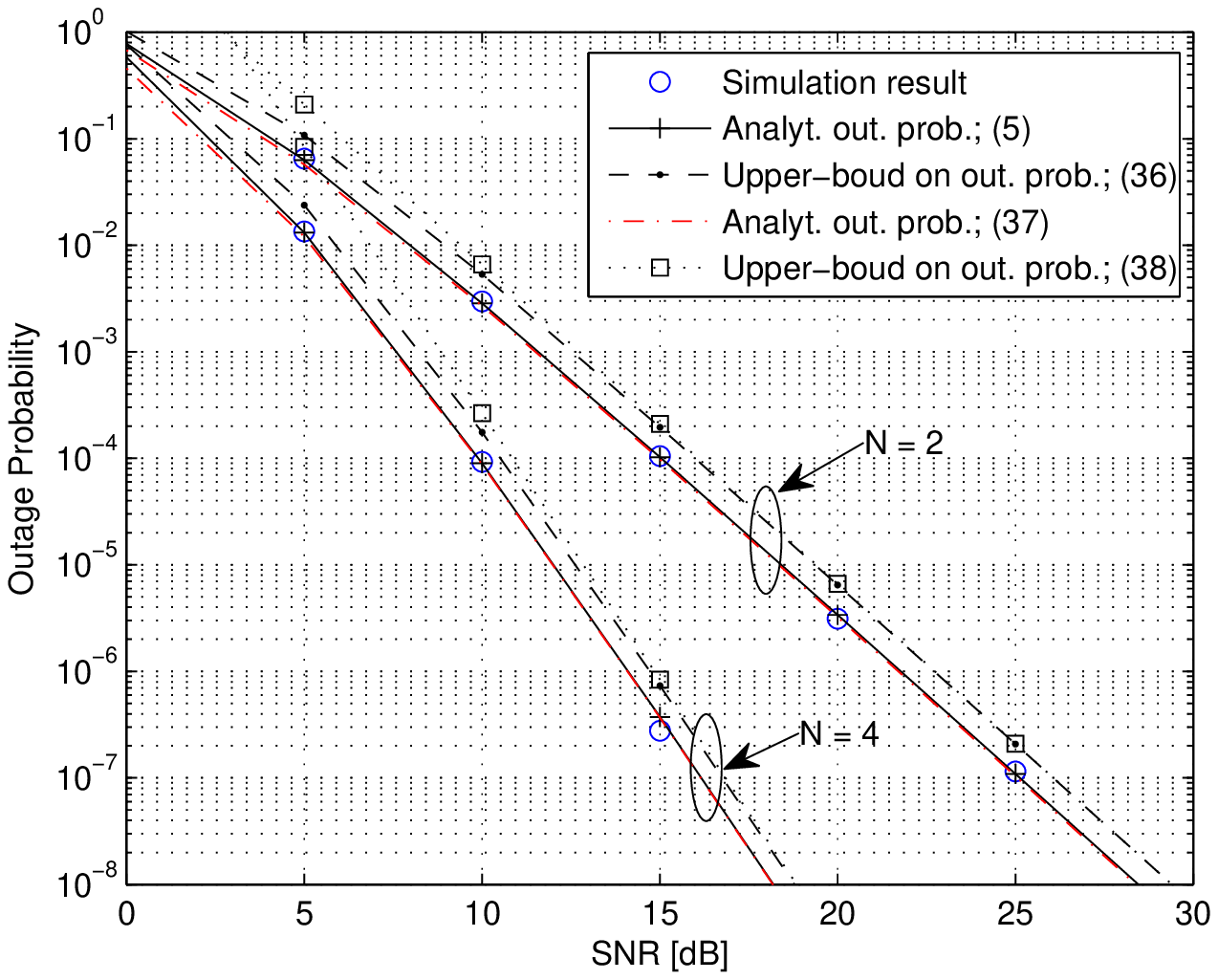}\\
  \vspace{-.3cm}
  \caption{The outage probability $P_{\text{out}}(l)$ curves of delay-limited HARQ networks employing opportunistic relaying with 2 and 4 relays, when $R=1$ bits/sec, $L=5$ is the maximum number of rounds, and we consider HARQ round of $l=2$.
  }\label{f3}
\end{figure}

It is straightforward to show that the outage probability for direct transmission after $l$ HARQ round is \cite[Eq. (7)]{nar08}
\begin{align}\label{61}
    &P_{\text{out},d}(l)= 1-\exp\left(-\frac{2^{R/l}-1}{\rho \, \sigma^2_{f_0}}\right).
\end{align}
In Fig.~\ref{g2}, the outage probability at the $l=L=5$th HARQ round of the system with different number of relays are considered. 
After selecting the best relay, Alamouti code is employed in the second transmission phase. Compared to the single HARQ relaying system proposed in \cite{nar08}, the proposed HARQ opportunistic relaying system with $N=2,3,4$ relays outperforms considerably for all SNR conditions. For example, it can be seen that in outage probability of $10^{-3}$, the system with two relays saves around $8$ dB in SNR compared to the single relay HARQ system.
Furthermore, it can be checked that the system with $N$ relays can achieve the diversity order of $N+1$.

A delay-limited throughput where defined in \eqref{b2} explicitly accounts for finite delay constraints and
associated non-zero packet outage probabilities. It can be shown that for small outage
probabilities, this delay-limited throughput is greater than the
conventional long-term average throughput defined in \eqref{b3}
In addition to finite delay constraint, represented by the
maximum number $L$ of HARQ rounds, higher-layer applications usually require that $P_{\text{out}}\leq \rho_{\max}$, where $\rho_{\max}$ is a target outage probability.
The total LT and DL throughput are studied in Fig.~\ref{g3} subject to user QoS constraints, represented by outage probability target $\rho_{\max}$ and delay constraint $L$.
In Fig.~\ref{g3}, the total LT and DL throughputs of opportunistic relying HARQ system with $N=2,4$ relays are plotted as a function of SNR and compared with the direct transmission HARQ system with $L=3$, $\rho_{\max}=10^{-3}$, and the following linear relay geometry: $\sigma^2_{f_r}=\sigma^2_{g_r}=\sigma^2_{0}=1$.
As expected, the presence of the relays significantly increases the throughput. Furthermore, in agreement with \cite[Eq. (5)]{nar08}, it can be seen that the delay-limited
throughput is greater than the long-term average throughput.
An interesting observation is that as well as the diversity gain achieved by the opportunistic relaying HARQ system, which is previously shown in Fig.~\ref{g2}, obtaining higher throughputs are possible.
This behavior underscores the importance of the proposed system.

\begin{figure}[e]
  \centering
  \vspace{-.3cm}
  \includegraphics[width=\columnwidth]{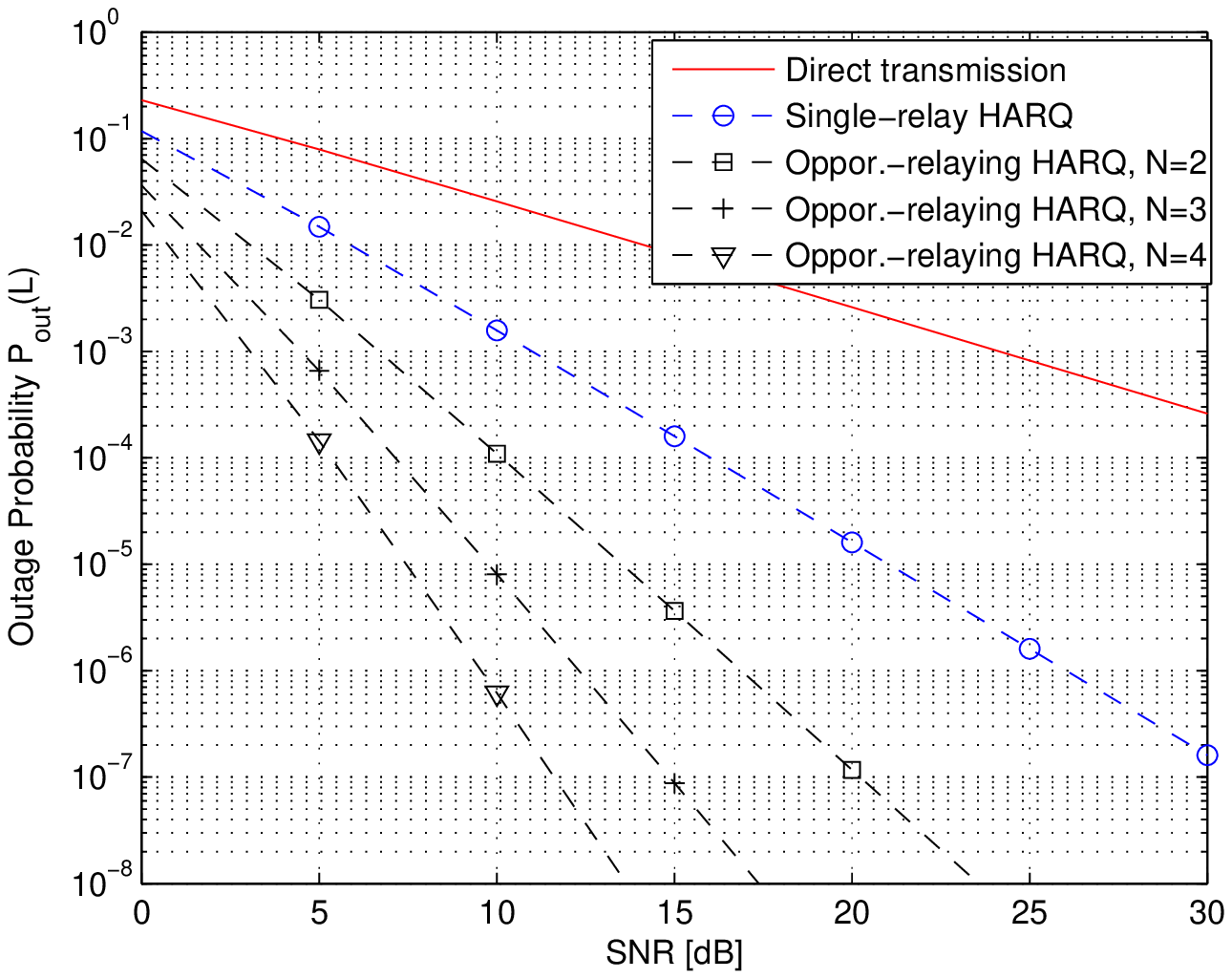}\\
  \vspace{-.3cm}
  \caption{The outage probability performance of the proposed relay selection HARQ system versus transmit SNR in a network with different number of relays, $R=1$ bits/sec, $L=l=5$ HARQ rounds, and $\sigma^2_{f_i}=\sigma^2_{g_i}=\sigma^2_{f_0}=1$.
  }\label{g2}
\end{figure}
\begin{figure}[e]
  \centering
  \vspace{-.3cm}
  \includegraphics[width=\columnwidth]{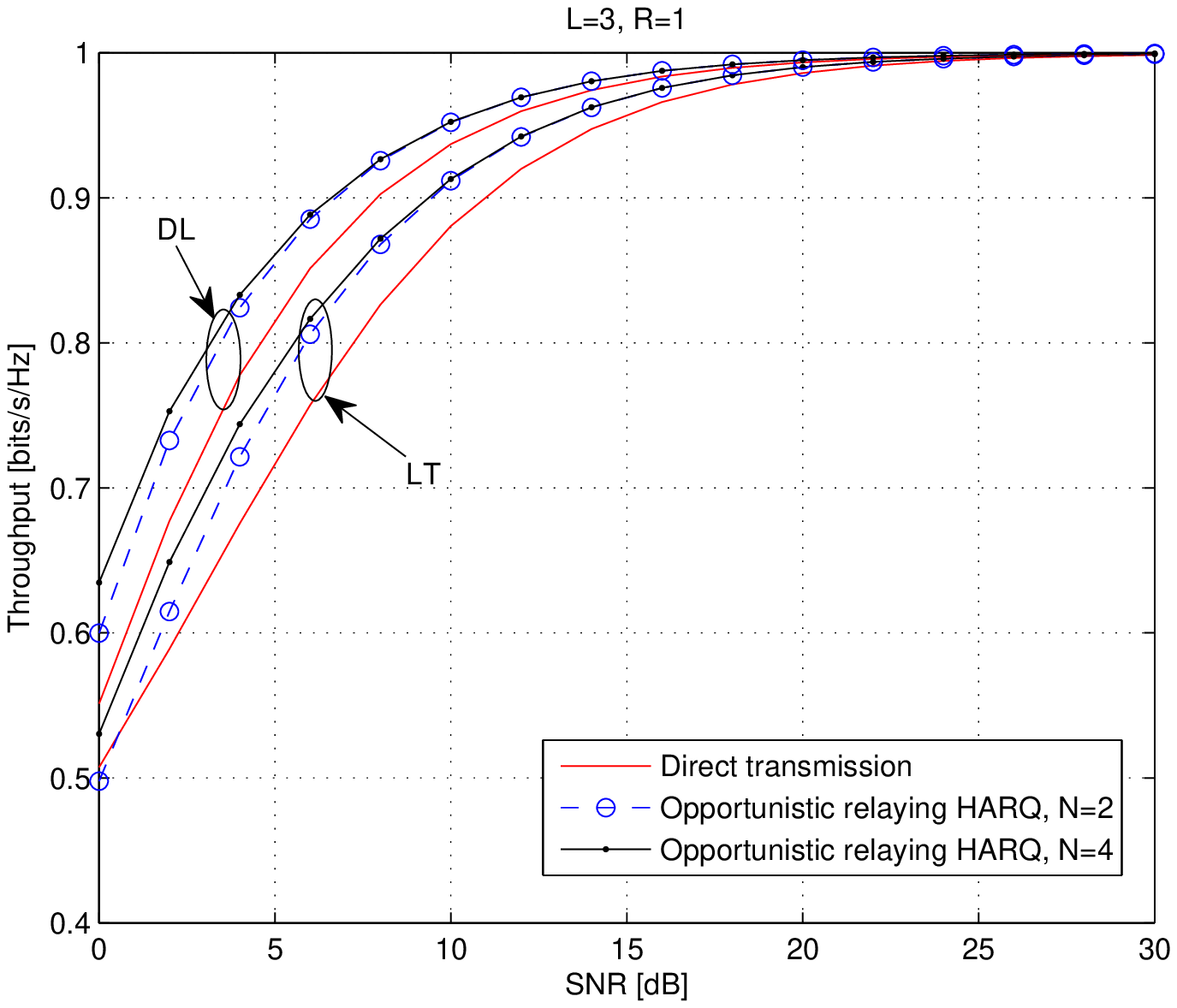}\\
  \vspace{-.3cm}
  \caption{ The delay-limited (DL) and long-term (LT) throughputs of direct transmission and relay selection HARQ system versus transmit SNR for target outage probability $\rho_{\max}=10^{-3}$, $L=3$ HARQ rounds, physical layer rate $R=1$ bits/sec, and $\sigma^2_{f_r}=\sigma^2_{g_r}=\sigma^2_{0}=1$, $\sigma^2_{f_i}=\sigma^2_{g_i}=\sigma^2_{f_0}=1$.
  }\label{g3}
\end{figure}

\section{Conclusion}
In this paper, we proposed a throughput-efficient relay selection HARQ system over Rayleigh fading. The throughput-delay performance of a half-duplex
multi-branch relay system with HARQ was analyzed.
A distributed relay selection scheme was introduced for HARQ multi-relay networks by using ACK/NACK signals transmitted by destination.
We evaluated the average throughput and outage error
probability performance and showed that the proposed technique significantly
reduces the multiplexing loss due to the half-duplex constraint
while providing attractive outage error probability performance.
The closed-form expressions outage probability were derived, defined as the probability of packet failure after $L$ HARQ rounds, in half-duplex. For sufficiently high SNR, we derived a simple closed-form average outage probability expression for a HARQ system with multiple cooperating branches.
Based on the derived upper-bound expressions, it was shown that the proposed scheme
achieves the full spatial diversity order of $N+1$ in a non-orthogonal relay network with $N$ parallel relays.
The analysis presented here allows quantitative evaluation of
the throughput-delay performance gain of the relay selection channel
compared to direct transmission.
The numerical results confirmed that the proposed schemes can bring diversity and multiplexing gains in the wireless relay networks.

\appendices

\section{Proof of Proposition 1}
First, we define the auxalary random variables $m_{i}\triangleq\min\{\gamma_{f_{i}},\gamma_{g_{i}}\}$ for $i=1,2,...,N$. Since $\{\gamma_{f_{i}},\gamma_{g_{i}}\}_{i=1}^{N}$ are independent exponential random variables, $m_{i}$s are also independent exponential random variables with the following CDF:
\begin{equation}\label{PE1}
\textnormal{F}_{m_{i}}(x)=1-e^{-(\frac{1}{\sigma^{2}_{f_{i}}}+\frac{1}{\sigma^{2}_{g_{i}}})x}\hspace{2mm},\hspace{2mm}x\geqslant 0
\end{equation}
Also, using partitioning theorem, we have
\begin{equation}\label{PE3}
\Pr\{\gamma_{f_{r}}\leqslant \gamma\}=\sum_{j=1}^{N}\Pr\{(\gamma_{f_{r}}\leqslant \gamma)\cap(r=j)\} .
\end{equation}
For $j=1,2,...,N$, the summands of \eqref{PE3} can be obtained as follows
\begin{align}\label{PE4}
\Pr\{(&\gamma_{f_{r}}\leqslant \gamma)\cap(r=j)\} \nonumber \\
&=\Pr\{(m_{j}\leqslant \gamma)\cap\big(\bigcap_{\underset{i\neq j}{i=1}}^{N}(m_{i}<m_{j})\big)\cap(\gamma_{f_{j}}\leqslant \gamma)\} \nonumber \\
&=\Pr\{(m_{j}\leqslant \gamma)\cap\big(\bigcap_{\underset{i\neq j}{i=1}}^{N}(m_{i}<m_{j})\big)\} \nonumber \\
&\hspace{3.3mm}-\Pr\{(m_{j}\leqslant \gamma)\cap\big(\bigcap_{\underset{i\neq j}{i=1}}^{N}(m_{i}<m_{j})\big)\cap(\gamma_{f_{j}}> \gamma)\} .
\end{align}

By substituting \eqref{PE4} in \eqref{PE3}, we obtain
\begin{align}\label{PE4.1}
\Pr\{\gamma_{f_{r}}&\leqslant \gamma\}=\sum_{j=1}^{N}\Pr\{(m_{j}\leqslant \gamma)\cap\big(\bigcap_{\underset{i\neq j}{i=1}}^{N}(m_{i}<m_{j})\big)\} \nonumber \\
&-\sum_{j=1}^{N} \Pr\{(m_{j}\leqslant \gamma)\cap\big(\bigcap_{\underset{i\neq j}{i=1}}^{N}(m_{i}<m_{j})\big)\cap(\gamma_{f_{j}}> \gamma)\}  \nonumber \\
&=\Pr\{\max(m_{1},m_{2},...,m_{N})\leqslant \gamma\} \nonumber \\
&-\sum_{j=1}^{N} \Pr\{(m_{j}\leqslant \gamma)\cap\big(\bigcap_{\underset{i\neq j}{i=1}}^{N}(m_{i}<m_{j})\big)\cap(\gamma_{f_{j}}> \gamma)\}  .
\end{align}
Since $m_{i}$s are independent, the first term on the right side of \eqref{PE4.1} is given by
\begin{align}\label{PE4.2}
\Pr\{&\max(m_{1},m_{2},...,m_{N})\leqslant \gamma\}=\Pr\{\bigcap_{i=1}^{N}(m_{i}\leqslant\gamma)\} \nonumber \\
&=\prod_{i=1}^{N}\Pr\{m_{i}\leqslant\gamma\}=\prod_{i=1}^{N}\textnormal{F}_{m_{i}}(\gamma) .
\end{align}

Also, the summand of the summation on the right side of  \eqref{PE4.1} is obtained as follows
\begin{align}\label{PE6}
\Pr\{(&m_{j}\leqslant \gamma)\cap\big(\bigcap_{\underset{i\neq j}{i=1}}^{N}(m_{i}<m_{j})\big)\cap(\gamma_{f_{j}}> \gamma)\} \nonumber \\
&=\Pr\{\gamma_{f_{j}}>\gamma\} \Pr\{(m_{j}\leqslant \gamma)\cap\big(\bigcap_{\underset{i\neq j}{i=1}}^{N}(m_{i}<m_{j})\big)|\gamma_{f_{j}}>\gamma\} \nonumber \\
&=\Pr\{\gamma_{f_{j}}>\gamma\} \Pr\{(\gamma_{g_{j}}\leqslant \gamma)\cap\big(\bigcap_{\underset{i\neq j}{i=1}}^{N}(m_{i}<\gamma_{g_{j}})\big)\}\nonumber \\
&=e^{-\frac{\gamma}{\sigma^{2}_{f_{j}}}}\int_{0}^{\gamma}\frac{1}{\sigma^{2}_{g_{j}}}e^{-\frac{\beta}{\sigma^{2}_{g_{j}}}}\Pr\{\bigcap_{\underset{i\neq j}{i=1}}^{N}(m_{i}<\beta)\}d\beta \nonumber \\
&=\frac{1}{\sigma^{2}_{g_{j}}}e^{-\frac{\gamma}{\sigma^{2}_{f_{j}}}}\int_{0}^{\gamma}e^{-\frac{\beta}{\sigma^{2}_{g_{j}}}}\prod_{\underset{i\neq j}{i=1}}^{N}\textnormal{F}_{m_{i}}(\beta)d\beta .
\end{align}
Substituting from \eqref{PE1} into \eqref{PE4.2} and \eqref{PE6}, one can obtain the CDF in \eqref{P1} using \eqref{PE4.1} to \eqref{PE6}. Also, taking derivative of \eqref{P1} with respect to $\gamma$, results in the PDF of $\gamma_{f_{r}}$, given by \eqref{P2}.

\section{Proof of Proposition 2}
For deriving the CDF of $\gamma_{f_r}$, we should first find the
CDF of  $\gamma_{\max}$, which can be written as
\begin{align}\label{10}
    \text{Pr}\{\gamma_{\max}<\gamma\}&=\text{Pr}\{\gamma_1<\gamma,\gamma_2<\gamma,\ldots,\gamma_N<\gamma\}
\end{align}
where $\gamma_i=\min\left\{\gamma_{f_i},\gamma_{g_i}\right\}$ is again
an exponential random variable (RV) with the parameter equal to
the sum of parameters of exponential RV $\gamma_{f_i}$ and
$\gamma_{g_i}$, i.e., $1/\sigma^2_{f_i}$ and $1/\sigma^2_{g_i}$, respectively.

Thus, assuming that all channel coefficients are independent of each others, we can rewrite \eqref{10} as
\begin{align}\label{11}
    \text{Pr}\{\gamma_{\max}<\gamma\}&=\prod_{i=1}^N\left(1-e^{-\gamma\left(\frac{1}{\sigma^2_{f_i}}
    +\frac{1}{\sigma^2_{g_i}}\right)}\right).
\end{align}
On the other hand, we have
\begin{align}\label{12}
    &\text{Pr}\{\gamma_{\max}<\gamma\}=1-\text{Pr}\{\min\left\{\gamma_{f_r},\gamma_{g_r}\right\}>\gamma\}
    \nonumber\\
    &=1-\text{Pr}\{\gamma_{f_r}>\gamma,\gamma_{g_r}>\gamma\}\approx 1-\text{Pr}\{\gamma_{f_r}\!>\!\gamma\}\text{Pr}\{\gamma_{g_r}\!>\!\gamma\},
\end{align}
where the last equality is an approximation as if $\gamma_{f_r}$ and $\gamma_{g_r}$ are independent.
For simplicity, we
assume equidistance source-relay and relay-destination links, i.e.,
that $\sigma^2_{f_i}=\sigma^2_{g_i}=\sigma^2_{i}$. Since we have
assumed that $\gamma_{f_i}$ and $\gamma_{g_i}$ have the same statistics,
using \eqref{11} and \eqref{12}, we have
\begin{align}\label{13}
    &\text{Pr}\{\gamma_{f_r}<\gamma\}=\text{Pr}\{\gamma_{g_r}<\gamma\}
    \approx 1-\sqrt{1-\prod_{i=1}^N\left(1-e^{-\frac{2\gamma}{\sigma^2_{i}}}\right)}.
\end{align}

\section{Proof of Proposition 3}
From a Taylor series expansion, it can be shown that the first term in \eqref{21p} is $O(1/\rho^{2N+1})$. From \eqref{21p}, and by representing the factor $\mu_{k}$ in terms of the SNR ratio $\rho$, the outage probability in high SNR can be written as
\begin{align}\label{60}
    &P_{\text{out}}(l)\leq \frac{\Delta(l)}{\rho^{N+1}},
\end{align}
where
$$
    \Delta(l)=\left(2^{\frac{R}{l}}-1\right)\left(2^{\frac{R}{l-1}}-1\right)^{\!N} \frac{L-l+1}{\sigma_{f_0}^2}
    \prod_{i=1}^N \! \left(\frac{1}{\sigma^2_{f_i}}+\frac{1}{\sigma^2_{g_i}}\right).
$$
Hence, observing \eqref{60}, the diversity order defined in \eqref{14h} is equal to
$N+1$, which is the full spatial diversity for $N+1$ transmitting nodes.

\bibliographystyle{ieeetr}
\bibliography{references}

\begin{thebibliography}{10}

\bibitem{nos04}
A.~Nosratinia, T.~Hunter, and A.~Hedayat, ``Cooperative communication in
  wireless networks,'' {\em IEEE Commun. Mag.}, vol.~42, no. 10, pp.~74--80,
  Oct. 2004.

\bibitem{cov79}
T.~Cover and A.~Gamal, ``Capacity theorems for the relay channel,'' {\em IEEE
  Trans. Info. Theory}, vol.~25, pp.~572--584, Sep. 1979.

\bibitem{lan04}
J.~N. Laneman, D.~Tse, and G.~Wornell, ``Cooperative diversity in wireless
  networks: Efficient protocols and outage behavior,'' {\em IEEE Trans. Inform.
  Theory}, vol.~50, no. 12, pp.~3062--3080, Dec. 2004.

\bibitem{mah09twc}
B.~Maham, A.~Hjørungnes, and G.~Abreu, ``Distributed {GABBA} space-time codes
  in amplify-and-forward relay networks,'' {\em IEEE Trans. Wireless Commun.},
  vol.~8, no. 4, pp.~2036--2045, Apr. 2009.

\bibitem{mah08nov}
B.~Maham and A.~Hjørungnes, ``Performance analysis of repetition-based
  cooperative networks with partial statistical {CSI} at relays,'' {\em IEEE
  Comm. Letters}, vol.~12, no. 11, pp.~828--830, Nov. 2008.

\bibitem{hun06}
T.~E. Hunter, S.~Sanayei, and A.~Nosratinia, ``Outage analysis of coded
  cooperation,'' {\em IEEE Trans. Inform. Theory}, vol.~52, no. 2,
  pp.~375--391, Feb. 2006.

\bibitem{kra05}
M.~G. G.~Kramer and P.~Gupta, ``Cooperative strategies and capacity theorems
  for relay networks,'' {\em IEEE Trans. Inform. Theory}, pp.~3037--3063, Sep.
  2005.

\bibitem{lin06}
Z.~Lin, E.~Erkip, and A.~Stefanov, ``Cooperative regions and partner choice in
  coded cooperative systems,'' {\em IEEE Trans. Commun.}, vol.~54, no. 7,
  pp.~1323--1334, Jul. 2006.

\bibitem{mah09eur}
B.~Maham and A.~Hjørungnes, ``Power allocation strategies for distributed
  space-time codes in amplify-and-forward mode,'' {\em EURASIP Journal on
  Advances in Signal Processing}, vol.~2009, Article ID 310247, 13 pages, 2009.

\bibitem{mah11tc}
B.~Maham, A.~Hjørungnes, and R.~Narasimhan, ``Energy-efficient space-time coded
  cooperation in outage-restricted multihop wireless networks,'' {\em IEEE
  Trans. Commun.}
\newblock to appear Nov. 2011, Available online
  http://persons.unik.no/behrouz/documents/TC2.pdf.

\bibitem{fan07}
Y.~Fan, C.~Wang, J.~Thompson, and H.~Poor, ``Recovering multiplexing loss
  through successive relaying using repetition coding,'' {\em IEEE Trans.
  Wireless Commun.}, vol.~6, no. 12, pp.~4484--4493, Dec. 2007.

\bibitem{tan08}
R.~Tannious and A.~Nosratinia, ``Spectrally efficient relay selection with
  limited feedback,'' {\em IEEE Journal Select. Areas Commun.}, vol.~26, no. 8,
  pp.~1419--1428, Oct. 2008.

\bibitem{nar08}
R.~Narasimhan, ``Throughput-delay performance of half-duplex hybrid-{ARQ} relay
  channels,'' in {\em Proc. IEEE Int. Conf. Commun. (ICC)}, (Beijing, China),
  May 2008.

\bibitem{qi09}
Y.~Qi, R.~Hoshyar, and R.~Tafazolli, ``On the performance of {HARQ} with hybrid
  relaying schemes,'' in {\em Proc. IEEE Intern. Conf. Commun. (ICC)},
  (Dresden, Germany), Jun. 2009.

\bibitem{mah09iet}
B.~Maham and A.~Hjørungnes, ``Differential space-time coded cooperation for
  decode-and-forward based wireless relay networks,'' {\em IET Communications
  (former IEE Proc. COM)}, vol.~4, no. 6, pp.~631--638, June 2010.

\bibitem{cai01}
G.~Caire and D.~Tuninetti, ``The throughput of hybrid-{ARQ} protocols for the
  {G}aussian collision channel,'' {\em IEEE Trans. Inf. Theory}, vol.~47,
  pp.~1971--1988, Jul. 2001.

\bibitem{tab05}
T.~Tabet, S.~Dusad, and R.~Knopp, ``Achievable diversity-multiplexing-delay
  tradeoff in half-duplex {ARQ} relay channels,'' in {\em Proc. IEEE ISIT},
  (Adelaide, Australia), Sep. 2005.

\bibitem{zha05jsac}
B.~Zhao and M.~C. Valenti, ``Practical relay networks: a generalization of
  hybrid-{ARQ},'' {\em IEEE J. Select. Areas Commun.}, vol.~23, no. 1,
  pp.~7--18, Jan. 2005.

\bibitem{ble06b}
A.~Bletsas, A.~Khisti, D.~P. Reed, and A.~Lippman, ``A simple cooperative
  method based on network path selection,'' {\em IEEE Journal on Selected Areas
  in Communications}, vol.~24, no. 3, pp.~659--672, Mar. 2006.

\bibitem{elg06}
H.~E. Gamal, G.~Caire, and M.~O. Damen, ``The {MIMO} {ARQ} channel:
  diversity-multiplexing-delay tradeoff,'' {\em IEEE Trans. Info. Theory},
  vol.~52, no. 8, pp.~3601--3621, Aug. 2006.

\bibitem{zor03a}
M.~Zorzi and R.~R. Rao, ``Geographic random forwarding ({G}e{R}a{F}) for ad hoc
  and sensor networks: {M}ultihop performance,'' {\em IEEE Trans. Mobile
  Comput.}, vol.~2, no. 4, pp.~337--348, Oct.-Dec. 2003.

\bibitem{gra96}
I.~S. Gradshteyn and I.~M. Ryzhik, {\em Table of {I}ntegrals, {S}eries, and
  {P}roducts}.
\newblock San Diego, USA: Academic, 1996.

\bibitem{jaf05}
H.~Jafarkhani, {\em Space-{T}ime {C}oding {T}heory and {P}ractice}.
\newblock Cambridge, UK: Cambridge Academic Press, 2005.

\bibitem{rib05}
A.~Ribeiro, A.~Cai, and G.~B. Giannakis, ``Symbol error probablity for general
  cooperative links,'' {\em IEEE Trans. Wireless Commun.}, vol.~4,
  pp.~1264--1273, May 2005.

\end{thebibliography}

\end{document}